\documentclass[runningheads,a4paper]{llncs}
\usepackage{amssymb}
\usepackage{graphicx}
\usepackage{url}
\usepackage{tcolorbox}
\usepackage{amsmath}
\usepackage{amsfonts}
\usepackage{verbatim}

\def\dosth#1{\ifx###1##\else\dofirst#1\anytoken\fi}
\def\doagain#1\anytoken{\dosth{#1}}
\def\payoffpairs#1#2#3{\m=#1\multiply\m by 4 \advance\m by -1 \n=1
  \def\dofirst##1{\put(\n,-\m){\makebox(0,0){\strut##1}}\advance\n by 4 \doagain}%
  \dosth{#2\strut}%
  \m=#1\multiply\m by 4 \advance\m by -3 \n=3 \dosth{#3\strut}}
\def\singlepayoffs#1#2{\m=#1\multiply\m by 4 \advance\m by -2 \n=2
  \def\dofirst##1{\put(\n,-\m){\makebox(0,0){\strut##1}}\advance\n by 4 \doagain}%
  {\large\dosth{#2\strut}}}
\newcommand{\bimatrixgame}[8]{%
\setlength{\unitlength}{#1}%
\newcount\rows
\newcount\cols
\rows=#2
\cols=#3
\newcount\rowcoord
\newcount\colcoord
\rowcoord=\rows
\colcoord=\cols
\multiply\rowcoord by 4
\multiply\colcoord by 4
\newcount\m
\newcount\n
\m=\rowcoord
\n=\colcoord
\advance\m by 2 
\advance\n by 2 
\begin{picture}(\n,\m)(-2,-\rowcoord)
\m=\rows
\n=\cols
\advance\m by 1
\advance\n by 1 
\thinlines
\multiput(0,0)(0,-4){\m}{\line(1,0){\colcoord}}
\multiput(0,0)(4,0){\n}{\line(0,-1){\rowcoord}}
\put(0,0){\line(-1,1){2}}
\put(-1.5,0.5){\makebox(0,0)[r]{#4}}  
\put(-.7,1.7){\makebox(0,0)[l]{#5}}   
\n=2
\def\dofirst##1{\put(-0.8,-\n){\makebox(0,0)[r]{\strut##1}}\advance\n by 4
   \doagain}
\dosth{#6\strut} 
\n=2
\def\dofirst##1{\put(\n,1.0){\makebox(0,0){\strut##1}}\advance\n by 4
   \doagain}
\dosth{#7\strut}#8%
\end{picture}}
\usepackage{cite}

\usepackage[disable]{todonotes}
\usepackage{todonotes}

\usepackage{xspace}

\newcommand{\eps}{\ensuremath{\epsilon}\xspace}
\newcommand{\supp}{\mathrm{supp}}

\newcommand{\profx}{\ensuremath{\mathbf{x}}\xspace}
\newcommand{\xs}{\ensuremath{\mathbf{x}^*}\xspace}
\newcommand{\ys}{\ensuremath{\mathbf{y}^*}\xspace}

\newcommand{\xu}{\ensuremath{\mathbf{x_b}}\xspace}
\newcommand{\xl}{\ensuremath{\mathbf{x_s}}\xspace}
\newcommand{\xmp}{\ensuremath{\mathbf{x_{mp}}}\xspace}
\newcommand{\ymp}{\ensuremath{\mathbf{y_{mp}}}\xspace}
\newcommand{\profy}{\ensuremath{\mathbf{y}}\xspace}

\newcommand{\js}{\ensuremath{\mathfrak{j}^*}\xspace}
\newcommand{\jp}{\ensuremath{\mathfrak{j}'}\xspace}

\newcommand{\pr}{\ensuremath{\mathrm{Pr}}\xspace}
\newcommand{\rr}{\ensuremath{\mathfrak{r}}\xspace}
\newcommand{\cc}{\ensuremath{\mathfrak{c}}\xspace}
\newcommand{\rh}{\ensuremath{\hat{\mathfrak{r}}}\xspace}

\newcommand{\regret}{\ensuremath{\mathcal{R}}\xspace}

\newcommand{\sr}{\ensuremath{S}\xspace}
\newcommand{\brs}{\ensuremath{B}\xspace}

\newcounter{alg}
\newcommand{\alg}[1]{\refstepcounter{alg}\label{#1}}

\begin{document}

\title{Distributed Methods for \\ Computing Approximate Equilibria}
\author{
Artur Czumaj\inst{1}
\and
Argyrios Deligkas\inst{2}
\and
Michail Fasoulakis\inst{1}
\and
John Fearnley\inst{2}
\and
Marcin Jurdzi\'nski\inst{1}
\and
Rahul Savani\inst{2}
}
\institute{
Department of Computer Science and DIMAP, University of Warwick, UK \and
Department of Computer Science, University of Liverpool, UK
}

\authorrunning{
Czumaj,
Deligkas,
Fasoulakis,
Fearnley,
Jurdzi\'nski, and
Savani.
}

\maketitle

\begin{abstract}
We present a new, distributed method to compute approximate Nash equilibria in bimatrix games. In contrast to previous approaches that analyze the two payoff matrices at the same time (for example, by solving a single LP that combines the two players payoffs), our algorithm first solves two independent LPs, each of which is derived from one of the two payoff matrices, and then compute approximate Nash equilibria using only limited communication between the players.

Our method has several applications for improved bounds for efficient computations of approximate Nash equilibria in bimatrix games.
First, it yields a best polynomial-time algorithm for computing \emph{approximate well-supported Nash equilibria (WSNE)}, which guarantees to find a 0.6528-WSNE in polynomial time.
Furthermore, since our algorithm solves the two LPs separately, it can be used to improve upon the best known algorithms in the limited communication setting: the algorithm can be implemented to obtain a randomized expected-polynomial-time algorithm that uses poly-logarithmic communication and finds a 0.6528-WSNE.
%
The algorithm can also be carried out to beat the best known bound in the query complexity setting, requiring $O(n \log n)$ payoff queries to compute a 0.6528-WSNE.
Finally, our approach can also be adapted to provide the best known communication efficient algorithm for computing \emph{approximate Nash equilibria}: it uses poly-logarithmic communication to find a 
0.382-approximate Nash equilibrium.
%
\end{abstract}

\section{Introduction}

The problem of finding equilibria in non-cooperative games is a central problem
in modern game theory. Nash's seminal theorem proved that every finite
normal-form game has at least one \emph{Nash equilibrium}~\cite{N51}, and this
raises the natural question of whether we can find one efficiently. After
several years of extensive research, this study has culminated in a proof that
finding a Nash equilibrium is PPAD-complete~\cite{DGP09} even for two-player
\emph{bimatrix games}~\cite{CDT09}, which is considered to be strong evidence
that there is no polynomial-time algorithm for this problem.

\paragraph{\bf Approximate equilibria.}
The fact that computing an exact Nash equilibrium of a bimatrix game is unlikely
to be tractable, has led to the study of \emph{approximate} Nash equilibria.
There are in fact two notions of approximate equilibrium, both of which will be
studied in this paper. An \emph{$\epsilon$-approximate Nash equilibrium}
($\epsilon$-NE) is a pair of strategies in which neither player can increase
their expected payoff by more than $\epsilon$ by deviating from their assigned
strategy. An \emph{$\epsilon$-well-supported Nash equilibrium} ($\epsilon$-WSNE)
is a pair of strategies in which both players only place probability on
strategies whose payoff is within $\epsilon$ of the best response payoff.
Every $\epsilon$-WSNE is an $\epsilon$-NE, but the converse does not hold, so a
WSNE is a more restrictive notion.

Approximate Nash equilibria are the more well studied of two concepts. A line of
work has studied the best guarantee that can be achieved in polynomial
time~\cite{DMP07,DMP09,BBM10}, and the best algorithm known so far is the
gradient descent method of Tsaknakis and Spirakis~\cite{TS08} that finds a $0.3393$-NE in
polynomial time. On the other hand, progress on computing
approximate-well-supported Nash equilibria has been less forthcoming. The first
correct algorithm was provided by Kontogiannis and Spirakis~\cite{KS10} (which shall
henceforth be referred to as the KS algorithm), who gave a polynomial time
algorithm for finding a $\frac{2}{3}$-WSNE. This was later slightly
improved by Fearnley, Goldberg, Savani, and S\o rensen~\cite{FGSS12} (whose
algorithm we shall refer to as the FGSS-algorithm), who showed
that the WSNEs provided by the KS algorithm could be improved, and this yields a
polynomial time algorithm for finding a $0.6608$-WSNE; this is the best
approximation guarantee for WSNEs that is currently known.

\paragraph{\bf Communication complexity.}

Approximate Nash equilibria can also be studied from the \emph{communication
complexity} point of view, which captures the amount of communication the
players need to find a good approximate Nash equilibrium. It models a natural
scenario where the two players each know their own payoff matrix, but do not
know their opponents payoff matrix. The players must then follow a communication
protocol that eventually produces strategies for both players. The goal is to
design a protocol that produces a sufficiently good $\epsilon$-NE or
$\epsilon$-WSNE while also minimizing the amount of communication between the
two players.

Communication complexity of equilibria in games has been studied in previous
works~\cite{CS04,HM10}. The recent paper of Goldberg and Pastink~\cite{GP14}
initiated the study of communication complexity in the bimatrix game setting.
There they showed $\Theta(n^2)$ communication is required to find an exact Nash
equilibrium of an $n \times n$ bimatrix game. Since these games have $\Theta(n^2)$
payoffs in total, this implies that there is no communication efficient protocol
for finding exact Nash equilibria in bimatrix games. For approximate equilibria,
they showed that one can find a $\frac{3}{4}$-Nash equilibrium \emph{without any
communication}, and that in the no-communication setting, finding an
$\frac{1}{2}$-Nash equilibrium is impossible. Motivated by these
positive and negative results, they focused on the most interesting setting,
which allows only a polylogarithmic (in $n$) amount of communication (number of bits) between the
players. They demonstrated that one can compute $0.438$-Nash equilibria and
$0.732$-well-supported Nash equilibria in this setting.

\paragraph{\bf Query complexity.}

The payoff query model is motivated by practical applications of game theory. It
is often the case that we know that there is a game to be solved, but we do not
know what the payoffs are, and in order to discover the payoffs, we would have
to play the game. This may be quite costly, so it is natural to ask whether we
can find an equilibrium of a game while minimising the number of experiments
that we must perform.

\emph{Payoff queries} model this situation. In the payoff query model we are
told the structure of the game, ie.\ the strategy space, but we are not told the
payoffs. We can then make payoff queries, where we propose a pure strategy
profile, and we are told the payoff to each player under that strategy profile.
Our task is to compute an equilibrium of the game while minimising the number of
payoff queries that we make.

The study of query complexity in bimatrix games was initiated by Fearnley et al.\ \cite{FGGS13}, who gave a deterministic algorithm
for finding a $\frac{1}{2}$-NE using $2n - 1$ payoff queries. A subsequent paper
of Fearnley and Savani~\cite{FS14} showed a number of further results. Firstly,
they showed a $\Omega(n^2)$ lower bound on the query complexity of finding an
$\epsilon$-NE with $\epsilon < \frac{1}{2}$, which combined with the result
above, gives a complete view of the deterministic query complexity of
approximate Nash equilibria in bimatrix games. They then give a randomized
algorithm that finds a $(\frac{3 - \sqrt{5}}{2} + \epsilon)$-NE using $O(\frac{n
\cdot \log n}{\epsilon^2})$ queries, and a randomized algorithm that finds a
$(\frac{2}{3} + \epsilon)$-WSNE using $O(\frac{n \cdot \log n}{\epsilon^4})$
queries.

\paragraph{\bf Our contribution.}

In this paper we introduce a \emph{distributed} technique that allows us to efficiently compute approximate Nash equilibria and approximate well-supported Nash equilibria using limited communication between the players.

Traditional methods for computing WSNEs have used an LP based approach that, when used on a bimatrix game $(R, C)$, solve the zero-sum game $(R - C, C - R)$. The KS algorithm \cite{KS10} showed that if there is no pure $\frac23$-WSNE, then the solution to the zero-sum game is a $\frac23$-WSNE. The slight improvement of the FGSS-algorithm \cite{FGSS12} to 0.6608 was obtained by adding two further methods to the KS algorithm: if the KS algorithm does not produce a 0.6608-WSNE, then either there is a $2 \times 2$ \emph{matching pennies} sub-game that is 0.6608-WSNE or the strategies from the zero-sum game can be improved by shifting the probabilities of both players within their supports in order to produce a 0.6608-WSNE.

In this paper, we take a different approach. We first show that the bound of $\frac23$ can be matched using a pair of \emph{distributed} LPs. Given a bimatrix game $(R, C)$, we solve the two zero-sum games $(R, -R)$, and $(-C, C)$, and we give a straightforward procedure that we call the \emph{base algorithm}, which uses the solutions to these games to produce a $\frac23$-WSNE of $(R, C)$. Goldberg and Pastink \cite{GP14} also considered this pair of LPs, but their algorithm only produces a 0.732-WSNE. We then show that the base algorithm can be improved by applying the probability-shifting and matching-pennies ideas from the FGSS-algorithm. That is, if the base algorithm fails to find a 0.6528-WSNE, then a 0.6528-WSNE can be obtained either by shifting the probabilities of one of the two players, or by identifying a $2 \times 2$ sub-game. This gives a polynomial-time algorithm that computes a 0.6528-WSNE, which provides the best known approximate guarantees for WSNEs (Theorem \ref{thm:main}).

It is worth pointing out that, while these techniques are thematically similar to the ones used by the FGSS-algorithm, the actual implementation is significantly different. The FGSS-algorithm attempts to improve the strategies by shifting probabilities \emph{within the supports} of the strategies returned by the two player game, with the goal of reducing the other player's payoff. In our algorithm, we shift probabilities \emph{away from bad strategies} in order to improve that player's payoff. This type of analysis is possible because the base algorithm produces a strategy profile in which one of the two players plays a pure strategy, which makes the analysis we need to carry out much simpler. On the other hand, the KS-algorithm can produce strategies in which both players play many strategies, and so the analysis used for the FGSS-algorithm is necessarily more complicated.

Since our algorithm solves the two LPs separately, it can be used to improve the upon the best known algorithms in the limited communication setting. This is because no communication is required for the row player to solve $(R, -R)$, and the column player to solve $(-C, C)$. The players can then carry out the rest of the algorithm using only poly-logarithmic communication. Hence, we obtain a randomized expected-polynomial-time algorithm that uses poly-logarithmic communication and finds a 0.6528-WSNE (Theorem \ref{thm:commbs}). Moreover, the base algorithm can be implemented as a communication efficient algorithm for finding a $(0.5 + \epsilon)$-WSNE in a \emph{win-lose} bimatrix game, where all payoffs are either 0 or~1 (Theorem \ref{thm:zo-comm}).

The algorithm can also used to beat the best known bound in the query complexity setting. It has already been shown by Goldberg and Roth~\cite{GR14} that an $\epsilon$-NE of a \emph{zero-sum game} can be found by a randomized algorithm that uses $O(\frac{n \log n}{\epsilon^2})$ payoff queries. Since rest of the steps used by our algorithm can also be carried out using $O(n \log n)$ payoff queries, this gives us a query efficient algorithm for finding a 0.6528-WSNE (Theorem \ref{thm:queries}).

We also show that the base algorithm can be adapted to find a $\frac{3 - \sqrt{5}}{2}$-NE in a bimatrix game, which matches the bound given for the first algorithm of Bosse et al.\ \cite{BBM10}. Once again, this can be implemented in a communication efficient manner, and so we obtain an algorithm that computes a $(\frac{3 - \sqrt{5}}{2} + \epsilon)$-NE (i.e., 0.382-NE) using only poly-logarithmic communication (Theorem \ref{thm:epscomm}).

\section{Preliminaries}
\label{sec:pre}

\paragraph{\bf Bimatrix games.}

Throughout the paper, we use $[n]$ to denote the set of integers $\{1, 2, \dots,
n\}$. An $n \times n$ bimatrix game is a pair $(R,C)$ of two $n \times n$
matrices: $R$ gives payoffs for the \emph{row} player, and $C$ gives the payoffs
for the \emph{column} player. We make the standard assumption that all payoffs
lie in the range $[0, 1]$. We also assume that each payoff has constant
bit-length.
A \emph{win-lose} bimatrix game is a
game in which all payoffs are either $0$ or $1$.

Each player has $n$ \emph{pure} strategies. To play the
game, both players simultaneously select a pure strategy: the row player selects
a row $i \in [n]$, and the column player selects a column $j \in [n]$. The row
player then receives payoff $R_{i,j}$, and the column player receives payoff
$C_{i,j}$.

A \emph{mixed strategy} is a probability distribution over $[n]$. We denote a
mixed strategy for the row player as a vector \profx of length $n$, such that
$\profx_i$ is the probability that the row player assigns to pure strategy $i$.
A mixed strategy of the column player is a vector \profy of length $n$, with the
same interpretation. Given a mixed strategy \profx for either player, the
\emph{support} of \profx is the set of pure strategies $i$ with $\profx_i >0$.
If \profx and \profy are mixed strategies for the row and the column player,
respectively, then we call $(\profx, \profy)$ a \emph{mixed strategy profile}.
The expected payoff for the row player under strategy profile $(\profx, \profy)$
is given by $\profx^T R \profy$ and for the column player by $\profx^T C
\profy$. We denote the \emph{support} of a strategy $\profx$ as $\supp(\profx)$,
which gives the set of pure strategies $i$ such that $\profx_i > 0$.

\paragraph{\bf Nash equilibria.}

Let $\profy$ be a mixed strategy for the column player. The set of \emph{pure
best responses} against $\profy$ for the row player is the set of pure
strategies that maximize the payoff against $\profy$. More formally, a pure
strategy $i \in [n]$ is a best response against $\profy$ if, for all pure
strategies $i' \in [n]$ we have: $\sum_{j \in \protect [n]} \profy_j \cdot R_{i,
j} \ge \sum_{j \in \protect [n]} \profy_j \cdot R_{i', j}$. Column player best
responses are defined analogously.

A mixed strategy profile $(\profx, \profy)$ is a \emph{mixed Nash equilibrium}
if every pure strategy in $\supp(\profx)$ is a best response against~$\profy$,
and every pure strategy in $\supp(\profy)$ is a best response against~$\profx$.
Nash~\cite{N51} showed that all bimatrix games have a mixed Nash equilibrium.
Observe that in a Nash equilibrium, each player's expected payoff is equal to
their best response payoff.

\paragraph{\bf Approximate Equilibria.}

There are two commonly studied notions of approximate equilibrium, and we
consider both of them in this paper. The first notion is of an
\emph{$\epsilon$-approximate Nash equilibrium} ($\epsilon$-NE), which weakens the
requirement that a player's expected payoff should be equal to their best
response payoff. Formally, given a strategy profile $(\profx, \profy)$, we
define the \emph{regret} suffered by the row player to be the difference between
the best response payoff, and the actual payoff:
\begin{equation*}
\max_{i\in [n]} \big((R\cdot y)_i\big) - \profx^T \cdot R \cdot \profy.
\end{equation*}
Regret for the column player is defined analogously. We have that $(\profx,
\profy)$ is an $\epsilon$-NE if and only if both players have regret less than
or equal to $\epsilon$.

The other notion is of an $\epsilon$-approximate-well-supported equilibrium
($\epsilon$-WSNE), which weakens the requirement that players only place
probability on best response strategies. Given a strategy profile $(\profx,
\profy)$ and a pure strategy $j \in [n]$, we say that $j$ is an
$\epsilon$-best-response for the row player if:
\begin{equation*}
\max_{i \in [n]}\big((R\cdot y)_i\big) - (R\cdot y)_j \le \epsilon.
\end{equation*}
An $\epsilon$-WSNE requires that both players only place probability on
$\epsilon$-best-responses. Formally,
the row player's \emph{pure strategy regret} under $(\profx, \profy)$ is defined
to be:
\begin{equation*}
\max_{i \in [n]}\big((R\cdot y)_i\big) -
\min_{i \in \supp(\profx)}\big((R\cdot y)_i\big).
\end{equation*}
Pure strategy regret for the column player is defined analogously. A strategy
profile $(\profx, \profy)$ is an $\epsilon$-WSNE if both players have pure
strategy regret less than or equal to $\epsilon$.

\paragraph{\bf Communication complexity.}

We consider the communication model for bimatrix games introduced by Goldberg and
Pastink~\cite{GP14}. In this model, both players know the payoffs in their own
payoff matrix, but do not know the payoffs in their opponent's matrix. The
players then follow an algorithm that uses a number of communication rounds,
where in each round they exchange a single bit of information. Between
each communication round, the players are permitted to perform arbitrary
randomized computations (although it should be noted that, in this paper, the players will
only perform polynomial-time computations) using their payoff matrix, and the
bits that they have received so far. At the end of the algorithm,
the row player
outputs a mixed strategy $\profx$, and the column player outputs a mixed
strategy $\profy$.

The goal is to produce a strategy profile $(\profx, \profy)$ that is
an $\epsilon$-NE or $\epsilon$-WSNE for a sufficiently small $\epsilon$ while
limiting the number of communication rounds used by the algorithm. The
algorithms given in this paper will use at most $O(\log^2 n)$ communication rounds.
In order to achieve this, we use the following result of Goldberg and
Pastink \cite{GP14}.

\begin{lemma}[\cite{GP14}]
\label{lem:GP14}
Given a mixed strategy $\profx$ for the row-player and an $\epsilon >
0$, there is a randomized expected-polynomial-time algorithm that uses
$O(\frac{\log^2 n}{\epsilon^2})$-communication to transmit a strategy $\profx_s$
to the column player where $|\supp(\profx_s)| \in O(\frac{\log n}{\epsilon^2})$
and for every strategy $i \in [n]$ we have:
\begin{equation*}
| (\profx^T \cdot R)_i - (\profx_s^T \cdot R)_i | \le \epsilon.
\end{equation*}
\end{lemma}
The algorithm uses the well-known sampling technique of Lipton, Markakis, and
Mehta to construct the strategy $\profx_s$, so for this reason we will call the
strategy $\profx_s$ the \emph{sampled strategy} from $\profx$. Since this
strategy has a logarithmically sized support, it can be transmitted by sending
$O(\frac{\log n}{\epsilon^2})$ strategy indexes, each of which can be
represented using $\log n$ bits. By symmetry, the algorithm can obviously also
be used to transmit approximations of column player strategies to the row
player.

\paragraph{\bf Query complexity.}

In the query complexity setting, the algorithm
knows that the players will play an $n \times n$ game $(R, C)$, but it does not
know any of the entries of $R$ or $C$. These payoffs are obtained using
\emph{payoff queries} in which the algorithm proposes a pure strategy profile
$(i, j)$, and then it is told the value of $R_{ij}$ and $C_{ij}$. After each
payoff query, the algorithm can make arbitrary computations (although, again, in
this paper the algorithms that we consider take polynomial time) in order to
decide the next pure strategy profile to query. After making a sequence of
payoff queries, the algorithm then outputs a mixed strategy profile $(\profx,
\profy)$. Again, the goal is to ensure that this strategy profile is an
$\epsilon$-NE or $\epsilon$-WSNE, while minimizing the number of queries made
overall.

There are two results that we will use for this setting. Goldberg and Roth~\cite{GR13} have
given a randomized algorithm that, with high probability, finds an $\epsilon$-NE
of a zero-sum game using $O(\frac{n \cdot \log n}{\epsilon^2})$ payoff
queries. Given a mixed strategy profile $(\profx, \profy)$, an
\emph{$\epsilon$-approximate payoff vector} for the row player is a vector $v$
such that, for all $i \in [n]$ we have $| v_i - (R \cdot \profy)_i | \le
\epsilon$. Approximate payoff vectors for the column player are defined
symmetrically. Fearnley and Savani~\cite{FS14} observed that there is a randomized algorithm
that when given the strategy profile $(\profx, \profy)$, finds approximate
payoff vectors for both players using $O(\frac{n \cdot \log n}{\epsilon^2})$
payoff queries and that succeeds with high probability. We summarise
these two results in the following lemma.

\begin{lemma}[\cite{GR13,FS14}]
\label{lem:GRFS}
Given an $n \times n$ zero-sum bimatrix game, with probability at least $(1 -
n^{-\frac{1}{8}})(1 - \frac{2}{n})^2$, we can compute an
$\epsilon$-Nash equilibrium $(\profx, \profy)$, and  $\epsilon$-approximate payoff
vectors for both players under $(\profx, \profy)$, using $O(\frac{n \cdot \log
n}{\epsilon^2})$ payoff queries.
\end{lemma}

\section{The base algorithm}

In this section, we introduce an algorithm that we call the \emph{base
algorithm}. This algorithm provides a simple way to find a $\frac{2}{3}$-WSNE.
We present this algorithm separately for three reasons. Firstly, we believe that
the algorithm is interesting in its own right, since it provides a relatively
straightforward method for finding a $\frac{2}{3}$-WSNE that is quite different
from the technique used in the KS-algorithm. Secondly, our algorithm for finding
a $0.6528$-WSNE will replace the final step of the algorithm with two more
involved procedures, so it is worth understanding this algorithm before we
describe how it can be improved. Finally, at the end of this section, we will
show that this algorithm can be adapted to provide a communication efficient way
to find a $(0.5 + \epsilon)$-WSNE in win-lose games.

\paragraph{\bf The algorithm.} Consider the following algorithm.

\alg{alg:base}
\begin{tcolorbox}[title=Algorithm~\ref{alg:base}]
\begin{enumerate}
\item
\label{itm:cfj-one}
Solve the zero-sum games $(R, -R)$ and $(-C, C)$.
\begin{itemize}
\item Let $(\profx^*, \profy^*)$ be
a NE of $(R, -R)$, and let $(\hat{\profx}, \hat{\profy})$ be a NE of $(C, -C)$.
\item
Let $v_r$ be the value secured by $\profx^*$ in $(R, -R)$, and let
$v_c$ be the value secured by $\hat{\profy}$ in $(-C, C)$. Without loss of
generality assume that $v_c \le v_r$.
\end{itemize}
\item
\label{itm:cfj-two}
If $v_r \leq 2/3$, then return $(\hat{\profx}, \profy^*)$.
\item
\label{itm:cfj-three}
If for all $j \in [n]$ it holds that $C^T_j \cdot \profx^* \leq 2/3$, then
return $(\profx^*, \profy^*)$.
\item
\label{itm:cfj-four}
Otherwise:
\begin{itemize}
\item Let $\js$ be a pure best response to $\profx^*$.
\item Find a row $i$ such that $R_{i\js} > 1/3$ and $C_{i\js} > 1/3$.
\item Return $(i, \js)$.
\end{itemize}
\end{enumerate}
\end{tcolorbox}

We argue that this algorithm is correct. Firstly, we must prove that the
row $i$ used in Step~\ref{itm:cfj-four} actually exists, which we do in the
following lemma.

\begin{lemma}
\label{lem:basepure}
If Algorithm~\ref{alg:base} reaches Step~\ref{itm:cfj-four},
then there exists a row $i$ such that $R_{i\js} > 1/3$ and $C_{i\js} > 1/3$.
\end{lemma}
\begin{proof}
Let $i$ be a row sampled from $\profx^*$. We will show that there is a positive
probability that row $i$ satisfies the desired properties.

We begin by showing that
the probability that $\pr(R_{i\js} \le \frac{1}{3}) < 0.5$. Let the random
variable $T = 1 - R_{i\js}$. Since $v_r > \frac{2}{3}$, we have that $E[T] <
\frac{1}{3}$. Thus, applying Markov's inequality we obtain:
\begin{equation*}
\pr(T \ge \frac{2}{3}) \le \frac{E[T]}{2/3} < 0.5.
\end{equation*}
Since $\pr(R_{i\js} \le \frac{1}{3}) = \pr(T \ge \frac{2}{3})$ we can
therefore conclude that $\pr(R_{i\js} \le \frac{1}{3}) < 0.5$. The exact same
technique can be used to prove that $\pr(C_{i\js} \le \frac{1}{3})
< 0.5$, by using the fact that $C^T_{\js} \cdot \profx^* > \frac{2}{3}$.

We can now apply the union bound to argue that:
\begin{equation*}
\pr(R_{i\js} \le \frac{1}{3} \text{ or } C_{i\js} \le \frac{1}{3}) < 1.
\end{equation*}
Hence, there is positive probability that row $i$ satisfies $R_{i\js} >
\frac{1}{3}$ and $C_{i\js} > \frac{1}{3}$, so such a row must exist. \qed
\end{proof}

We now argue that the algorithm always produces a $\frac{2}{3}$-WSNE. There are
three possible strategy profiles that can be returned by the algorithm, which we
consider individually.
\begin{description}
\item[Step~\ref{itm:cfj-two}.] Since $v_c \le v_r$ by assumption, and since $v_r \le
\frac{2}{3}$, we have that $(R \cdot \profy^*)_i \le \frac{2}{3}$ for
every row $i$, and $((\hat{\profx})^T \cdot C)_j \le \frac{2}{3}$ for every
column~$j$. So, both players can have pure strategy regret at most $\frac{2}{3}$ in $(\hat{\profx}, \profy^*)$, and thus this profile is a $\frac{2}{3}$-WSNE.
\item[Step~\ref{itm:cfj-three}.]
Much like in the previous case, when the column
player plays $\profy^*$, the row player can have pure strategy regret at most
$\frac{2}{3}$. The requirement that
$C^T_j\profx^* \leq \frac{2}{3}$ also ensures that the column player has
pure strategy regret at most $\frac{2}{3}$. Thus, we have that $(\profx^*,
\profy^*)$ is a $\frac{2}{3}$-WSNE.
\item[Step~\ref{itm:cfj-four}.] Both players have payoff at least $\frac{1}{3}$
under $(i, \js)$ for the sole strategy in their respective supports. Hence, the
maximum pure strategy regret that can be suffered by a player is $1 -
\frac{1}{3} = \frac{2}{3}$.
\end{description}

\noindent Therefore, we have show that the algorithm always produces a
$\frac{2}{3}$-WSNE.

\paragraph{\bf Win-lose games.}

The base algorithm can be adapted to provide a communication efficient method
for finding a $(0.5 + \epsilon)$-WSNE in win-lose games. In brief, the algorithm
can be modified to find a $0.5$-WSNE in a win-lose game by making
Steps~\ref{itm:cfj-two} and~\ref{itm:cfj-three} check against the threshold of
$0.5$. It can then be shown that if these steps fail, then there exists a pure
Nash equilibrium in column $\js$. This can then be implemented as a
communication efficient protocol using the algorithm from Lemma~\ref{lem:GP14}.
Full details are given in Appendix~\ref{app:zero-one}, where the following
theorem is proved.

\begin{theorem}
\label{thm:zo-comm}
For every win-lose game and every $\epsilon > 0$, there is a randomized
expected-polynomial-time algorithm that uses $O\left(\frac{\log^2
n}{\epsilon^2}\right)$ communication and finds a $(0.5 + \epsilon)$-WSNE.
\end{theorem}

\section{An algorithm for finding a $0.6528$-WSNE}

In this section, we show how Algorithm~\ref{alg:base} can be modified to
produce a $0.6528$-WSNE. We begin by giving an overview of the techniques used,
we then give the algorithm, and finally we analyse the quality of WSNE that it
produces.

\paragraph{\bf Outline.}
The idea behind our algorithm is to replace Step~\ref{itm:cfj-four} of
Algorithm~\ref{alg:base} with a more involved procedure. This procedure uses two
techniques, that both find an $\epsilon$-WSNE with $\epsilon < \frac{2}{3}$.

Firstly, we attempt to turn $(\profx^*, \js)$ into a WSNE by
\emph{shifting probabilities}. Observe that, since $\js$ is a best response, the
column player has a pure strategy regret of $0$ in $(\profx^*, \js)$. On the
other hand, we have no guarantees about the row player since $\profx^*$ might
place a small amount of probability strategies with payoff strictly less than
$\frac{1}{3}$. However, since $\profx^*$ achieves a high \emph{expected} payoff
(due to Step~\ref{itm:cfj-two},)
it cannot place too much probability on these low payoff strategies. Thus, the
idea is to shift the probability that $\profx^*$ assigns to entries of $\js$
with payoff less than or equal to $\frac{1}{3}$ to entries with payoff strictly
greater than $\frac{1}{3}$, and thus ensure that the row player's pure strategy
regret is below $\frac{2}{3}$. Of course, this procedure will increase the pure
strategy regret of the column player, but if it is also below $\frac{2}{3}$ once
all probability has been shifted, then we have found an $\epsilon$-WSNE with
$\epsilon < \frac{2}{3}$.

If shifting probabilities fails to find an $\epsilon$-WSNE with $\epsilon
< \frac{2}{3}$, then we show that the game contains a \emph{matching pennies}
sub-game. More precisely, we show that there exists a column $\jp$, and rows $b$
and $s$ such that the $2 \times 2$ sub-game induced by $\js$, $\jp$, $b$, and
$s$ has the following form:

\begin{center}
\bimatrixgame{3.7mm}{2}{2}{I}{II}%
{{$b$}{$s$}}%
{{$\js$}{$\jp$}}%
{
\payoffpairs{1}{{$\approx1$}{$0$}}{{$0$}{$\approx 1$}}
\payoffpairs{2}{{$0$}{$\approx 1$}}{{$\approx 1$}{$0$}}
}
\end{center}
Thus, if both players play uniformly over their respective pair of strategies,
then $\js$, $\jp$, $b$, and $s$ with have payoff $\approx 0.5$, and so this
yields an $\epsilon$-WSNE with $\epsilon < \frac{2}{3}$.

\paragraph{\bf The algorithm.}
We now formalize this approach, and show that it always finds an $\epsilon$-WSNE
with $\epsilon < \frac{2}{3}$. In order to quantify the precise $\epsilon$ that
we obtain, we parametrise the algorithm by a variable~$z$, which we constrain to
be in the range $0 \le z < \frac{1}{24}$. With the exception of the matching
pennies step, all other steps of the algorithm will return a $(\frac{2}{3} -
z)$-WSNE, while the matching pennies step will return a $(\frac{1}{2} +
f(z))$-WSNE for some increasing function $f$. Optimizing the trade off between
$\frac{2}{3} - z$ and $\frac{1}{2} + f(z)$ then allows us to determine the
quality of WSNE found by our algorithm.

The algorithm is displayed as Algorithm~\ref{alg:ws}. Observe that Steps~\ref{itm:one},~\ref{itm:two},
and~\ref{itm:three} are versions of the corresponding steps from
Algorithm~\ref{alg:base}, which have been adapted to produce a
$(\frac{2}{3}-z)$-WSNE. Step~\ref{itm:four} implements the probability shifting
procedure, while Step~\ref{itm:five} finds a matching pennies sub-game.

\alg{alg:ws}
\begin{figure}
\begin{tcolorbox}[title=Algorithm~\ref{alg:ws}]
\begin{enumerate}
\itemsep2mm
\item
\label{itm:one}
Solve the zero-sum games $(R, -R)$ and $(-C, C)$.
\begin{itemize}
\item Let $(\profx^*, \profy^*)$ be
a NE of $(R, -R)$, and let $(\hat{\profx}, \hat{\profy})$ be a NE of $(C, -C)$.
\item
Let $v_r$ be the value secured by $\profx^*$ in $(R, -R)$, and let
$v_c$ be the value secured by $\hat{\profy}$ in $(-C, C)$. Without loss of
generality assume that $v_c \le v_r$.
\end{itemize}
\item
\label{itm:two}
If $v_r \leq 2/3 - z$, then return $(\hat{\profx}, \profy^*)$.
\item
\label{itm:three}
If for all $j \in [n]$ it holds that $C^T_j\profx^* \leq 2/3 - z$, then
return $(\profx^*, \profy^*)$.
\item
\label{itm:four}
Otherwise:
\begin{itemize}
\item Let $\js$ be a pure best response against $\profx^*$. Define:
\begin{align*}
\sr  &:= \{i \in \supp(\xs): R_{i\js} < 1/3 + z\}  \\
\brs &:= \supp(\xs) \setminus \sr
\end{align*}
\item Define the strategy $\xu$ as follows. For each $i \in [n]$ we have:
\begin{equation*}
(\xu)_i = \begin{cases}
\frac{1}{\pr(\brs)} \cdot \xs_i & \text{if $i \in \brs$} \\
0 & \text{otherwise.}
\end{cases}
\end{equation*}
\item If $(\xu^T \cdot C)_{\js} \ge \frac{1}{3} + z$, then return $(\xu, \js)$.
\end{itemize}
\item \label{itm:five}
Otherwise:
\begin{itemize}
\item Let $\jp$ be a pure best response against $\xu$.
\item
If there exists an $i \in \supp(\xs)$ such that $(i, \js)$ or $(i, \jp)$ is a
pure $(\frac{2}{3}-z)$-WSNE, then return it.
\item Find a row $b \in \brs$ such that $R_{b\js} > 1- \frac{18z}{1 + 3z}$ and
$C_{b\jp} > 1- \frac{18z}{1 + 3z}$.
\item
Find a row $s \in \sr$ such that $C_{s\js} > 1- \frac{27z}{1 + 3z}$ and
$R_{s\jp} > 1- \frac{27z}{1 + 3z}$.
\item
Define the row player strategy $\xmp$ and the column player strategy $\ymp$ as
follows. For each $i \in [n]$ we have:
\begin{align*}
\xmp_i = \begin{cases}
\frac{1 - 24z}{2 - 39z} & \text{if $i = b$,} \\
\frac{1 - 15z}{2 - 39z} & \text{if $i = s$,} \\
0 & \text{otherwise.}
\end{cases}& &
\ymp_i = \begin{cases}
\frac{1 - 24z}{2 - 39z} & \text{if $i = \js$,} \\
\frac{1 - 15z}{2 - 39z} & \text{if $i = \jp$,} \\
0 & \text{otherwise.}
\end{cases}
\end{align*}
\item Return $(\xmp, \ymp)$.
\end{itemize}
\end{enumerate}
\end{tcolorbox}
\end{figure}

Observe that the probabilities used in $\xmp$ and $\ymp$ are only well
defined when $z \le \frac{1}{24}$, because we have that $\frac{1-15z}{2-39z} >
1$ whenever $z > \frac{1}{24}$, which explains our required upper bound on $z$.

\paragraph{\bf The correctness of Step~\ref{itm:five}.}
This step of the algorithm relies on the existence of the rows $b$ and $s$,
which is not at all trivial. This is shown in the following lemma. The proof of
this lemma is quite lengthy, and is given in full detail in
Appendix~\ref{app:bs}.

\begin{lemma}
\label{lem:bs}
Suppose that the following three conditions hold:
\begin{enumerate}
\item
\label{itm:pre-one}
$\profx^*$ has payoff at least $\frac{2}{3} - z$ against $\js$.
\item
\label{itm:pre-two}
$\js$ has payoff at least $\frac{2}{3} - z$ against $\profx^*$.
\item
\label{itm:pre-three}
$\profx^*$ has payoff at least $\frac{2}{3} - z$ against $\jp$.
\item
\label{itm:pre-four}
Neither $\js$ or $\jp$ contains a pure $(\frac{2}{3} - z)$-WSNE $(i, j)$
with $i \in \supp(\profx^*)$.
\end{enumerate}
Then, both of the following are true:
\begin{itemize}
\item There exists a row $b \in \brs$ such that $R_{b\js} > 1- \frac{18z}{1 + 3z}$ and
$C_{b\jp} > 1- \frac{18z}{1 + 3z}$.
\item
There exists a row $s \in \sr$ such that $C_{s\js} > 1- \frac{27z}{1 + 3z}$ and
$R_{s\jp} > 1- \frac{27z}{1 + 3z}$.
\end{itemize}
\end{lemma}

The lemma explicitly states the preconditions that need to hold because we will
reuse it in our communication complexity and query complexity results. Observe
that the preconditions are indeed true if the Algorithm reaches
Step~\ref{itm:five}. The first and third conditions hold because, due to
Step~\ref{itm:two}, we know that $\profx^*$ is a min-max strategy that secures
payoff at least $v_r > \frac{2}{3} - z$. The second condition holds because
Step~\ref{itm:three} ensures that the column player's best response payoff is at
least $\frac{2}{3} - z$. The fourth condition holds because Step~\ref{itm:five}
explicitly checks for these pure strategy profiles.

\paragraph{\bf Overview of the proof of Lemma~\ref{lem:bs}.}
We now give an overview of the ideas used in the proof.
The majority of the proof is dedicated to proving four facts, which we outline
below.
First we determine the
structure of the row $\js$. Here we use the fact that in $(\profx^*, \js)$ both
players have expected payoff close to $\frac{2}{3}$, but there does not exist a
row $i \in \supp(\profx^*)$ such that $R_{i\js} \ge \frac{1}{3} + z$ and
$C_{i\js} \ge \frac{1}{3} + z$ (because such a row would constitute a pure
$(\frac{2}{3} - z)$-WSNE.) The only way this is possible is both of the
following facts hold.
\begin{enumerate}
\item Most of the probability assigned to $B$ is placed on rows $i$ with $R_{i\js} \approx 1$ and $C_{i\js} \approx
\frac{1}{3}$.
\item Most of the probability assigned to $S$ is placed on rows $i$ with $R_{i\js} \approx \frac{1}{3}$ and $C_{i\js} \approx 1$.
\end{enumerate}
Moreover, $\profx^*$ must assign roughly half of its probability to rows in $B$
and half of its probability to rows in $S$.

Next, we observe that since Step~\ref{itm:four} failed to produce a
$(\frac{2}{3} - z)$-WSNE, it must be the case that $\js$ is not a $(\frac{2}{3}
- z)$-best-response against $\xu$, and the payoff of $\js$ against $\xu$ is
  approximately $\frac{1}{3}$, it must be the case that the payoff of $\jp$
against $\xu$ is close to $1$. The only way this is possible is if most column
player payoffs for rows in $B$ are close to $1$. However, if this is the case,
then since $\js$ does not contain a pure $(\frac{2}{3} - z)$-WSNE, we have that
most row player payoffs in $B$ must be below $\frac{1}{3} + z$. This gives us
our third fact.
\begin{enumerate}
\setcounter{enumi}{2}
\item Most of the probability assigned to $B$ is placed on rows $i$ with
$R_{i\jp} < \frac{1}{3}+z$ and $C_{i\jp} \approx 1$.
\end{enumerate}

For the fourth fact, we recall that $\profx^*$ is a min-max strategy that
guarantees payoff at least $v_r > \frac{2}{3} - z$, so the payoff of $\xs$
against $\jp$ must be at least $\frac{2}{3} - z$. However, since most rows $i
\in B$ have $R_{i\jp} < \frac{1}{3}+z$, and since $\xs$ places roughly half its
probability on $B$, it must be the case that most row player payoffs in $S$ are
close to $1$. This gives us our final fact.
\begin{enumerate}
\setcounter{enumi}{3}
\item Most of the probability assigned to $S$ is placed on rows $i$ with
$R_{i\jp} \approx 1$.
\end{enumerate}

Our four facts only describe the \emph{expected} payoff of the rows in $B$ and
$S$ for the columns $\js$ and $\jp$. The final step of the proof is to pick out
two particular rows that satisfy the desired properties. For the row
$b$ we use Facts~1 and~3, observing that if most of the probability assigned to
$B$ is placed on rows $i$ with $R_{i\js} \approx 1$, and on rows $i$ with
$C_{i\js} \approx 1$, then it must be the case that both of these conditions can
be simultaneously satisfied by a single row $b$. The existence of $s$ is proved
by the same argument using Facts~2 and~4.

\paragraph{\bf Quality of approximation.}
We now analyse the quality of WSNE that our algorithm produces.
Steps~\ref{itm:two},~\ref{itm:three},~\ref{itm:four},~\ref{itm:five} each return
a strategy profile. Observe that Steps~\ref{itm:two} and~\ref{itm:three} are the same
as the respective steps in the base algorithm, but with the threshold changed
from $\frac{2}{3}$ to $\frac{2}{3} - z$. Hence, we can use the same reasoning as
we gave for the base algorithm to argue that these steps always return
$(\frac{2}{3} - z)$-WSNE. We now consider the other two steps.
\begin{description}
\item[Step~\ref{itm:four}.] By definition all rows $r \in \brs$ satisfy
$R_{i\js} \ge \frac{1}{3} + z$, so since $\supp(\xu) \subseteq \brs$, the pure
strategy regret of the row player can be at most $1 - (\frac{1}{3} + z) =
\frac{2}{3} - z$. For the same reason, since $(\xu^T \cdot C)_{\js} \ge
\frac{1}{3} + z$ holds, the pure strategy regret of the column player can also
be at $\frac{2}{3} - z$. Thus, the profile $(\xu, \js)$ is a $(\frac{2}{3} -
z)$-WSNE.
\item[Step~\ref{itm:five}.]
Since $R_{b\js} > 1- \frac{18z}{1 + 3z}$, the payoff of $b$ when the column
player plays $\ymp$ is at least:
\begin{equation*}
\frac{1 - 24z}{2 - 39z} \cdot \left( 1 - \frac{18z}{1 + 3z} \right) =
\frac{1 - 39z + 360z^2}{2 - 33z - 117 z^2}
\end{equation*}

Similarly, since $R_{s\jp} > 1- \frac{27z}{1 + 3z}$, the payoff of $s$ when the
column player plays $\ymp$ is at least:
\begin{equation*}
\frac{1 - 15z}{2 - 39z} \cdot \left( 1- \frac{27z}{1 + 3z} \right) =
\frac{1 - 39z + 360z^2}{2 - 33z - 117 z^2}
\end{equation*}
In the same way, one can show that the payoffs of $\js$ and $\jp$ are also
$\frac{1 - 39z + 360z^2}{2 - 33z - 117 z^2}$ when the row player plays $\xmp$.
Thus, we have that $(\xmp, \ymp)$ is a $(1 - \frac{1 - 39z + 360z^2}{2 - 33z -
117 z^2})$-WSNE.
\end{description}

\noindent To find the optimal value for $z$, we need to find the largest value
of $z$ for which the following inequality holds.
\begin{equation*}
1 - \frac{1 - 39z + 360z^2}{2 - 33z - 117 z^2} \le \frac{2}{3} - z.
\end{equation*}
Setting the inequality to an equality and rearranging gives the following cubic
polynomial equation.
\begin{equation*}
117 \, z^{3} + 432 \, z^{2} - 30 \, z + \frac{1}{3} = 0.
\end{equation*}
Since the discriminant of this polynomial is positive, this polynomial has three
real roots, which can be found via the trigonometric method. Only one of
these roots lies in the range $0 \le z < \frac{1}{24}$, which is the following:
\begin{multline*}
z = \frac{1}{117} \, \sqrt{3} \Bigg(\sqrt{2434} \sqrt{3} \cos\left(\frac{1}{3} \,
\arctan\left(\frac{39}{240073} \, \sqrt{9749} \sqrt{3}\right)\right)\\
- 3 \,
\sqrt{2434} \sin\left(\frac{1}{3} \, \arctan\left(\frac{39}{240073} \,
\sqrt{9749} \sqrt{3}\right)\right) - 48 \, \sqrt{3}\Bigg).
\end{multline*}
Thus, we get $z \approx 0.013906376$, and so we have found an algorithm that
always produces a $0.6528$-WSNE. So we have the following theorem.

\begin{theorem}
\label{thm:main}
There is a polynomial time algorithm that, given a bimatrix game, finds a
$0.6528$-WSNE.
\end{theorem}

\paragraph{\bf Communication complexity.}

We claim that our algorithm can be adapted for the limited communication
setting. We make the following modification to our algorithm. After computing
$\xs, \ys, \hat{\profx},$ and $\hat{\profy}$, we then use Lemma~\ref{lem:GP14}
to construct and communicate the sampled strategies $\xs_s, \ys_s,
\hat{\profx}_s,$ and $\hat{\profy}_s$. These strategies are communicated between
the two players using $4 \cdot (\log n)^2$ bits of communication, and the
players also exchange $v_r = (\xs_s)^T \cdot R\ys_s$ and $v_c = \hat{\profx}_s^TC\hat{\profy}_s$ using $\log n$ rounds of communication.
\todo[inline]{Can $v_r$ actually be represented using $\log n$ bits?}
The algorithm then continues as before, except the sampled strategies are used
in place of their non-sampled counterparts. Finally, in Steps~\ref{itm:two}
and~\ref{itm:three}, we test against the threshold $\frac{2}{3} - z + \epsilon$
instead of $\frac{2}{3} - z$.

Observe that, when sampled strategies are used, all steps of the algorithm can
be carried out in at most $(\log n)^2$ communication. In particular, to
implement Step~\ref{itm:four}, the column player can communicate $\js$ to the
row player, and then the row player can communicate $R_{i\js}$ for all rows $i
\in \supp(\xs_s)$ using $(\log n)^2$ bits of communication, which allows the
column player to determine $\jp$. Once $\jp$ has been determined, there are only
$2 \cdot \log n$ payoffs in each matrix that are relevant to the algorithm (the
payoffs in rows $i \in \supp(\xs_s)$ in columns $\js$ and $\jp$,) and so the two
players can communicate all of these payoffs to each other, and then no further
communication is necessary.

Now, we must argue that this modified algorithm is correct. Firstly, we argue
that if the modified algorithm reaches Step~\ref{itm:five}, then the rows $b$
and $s$ exist. To do this, we observe that the required preconditions of
Lemma~\ref{lem:bs} are satisfied by $\profx^*_s$, $\js$, and $\jp$.
Condition~\ref{itm:pre-two} holds because the modified Step~\ref{itm:three}
ensures that the column player's best response payoff is at least $\frac{2}{3} -
z + \epsilon > \frac{2}{3} - z$, while Condition~\ref{itm:pre-four} is ensured
by the explicit check in Step~\ref{itm:five}. For Conditions~\ref{itm:pre-one}
and~\ref{itm:pre-three}, we use the fact that $(\profx^*, \profy^*)$ is an
$\epsilon$-Nash equilibrium of the zero-sum game $(R, -R)$. The following lemma
shows that any approximate Nash equilibrium of a zero-sum game behaves like an
approximate min-max strategy.

\begin{lemma}
\label{lem:apx-min-max}
If $(\profx, \profy)$ is an $\epsilon$-NE of a zero-sum game $(M, -M)$, then
for every strategy $\profy'$ we have:
\begin{equation*}
\profx^T \cdot M \cdot \profy' \ge \profx^T \cdot M \cdot \profy - \epsilon.
\end{equation*}
\end{lemma}
\begin{proof}
Let $v = \profx^T \cdot M \cdot \profy$ be the payoff to the row player under
$(\profx, \profy)$. Suppose, for the sake of contradiction, that there exists a
column player strategy $\profy'$ such that:
\begin{equation*}
\profx^T \cdot M \cdot \profy' < v - \epsilon.
\end{equation*}
Since the game is zero-sum, this implies that the column player's payoff under
$(\profx, \profy')$ is strictly larger than $-v + \epsilon$, which then directly
implies that
the best response payoff for the column player against $\profx$ is strictly
larger than $-v + \epsilon$. However, since the column player's expected payoff
under $(\profx, \profy)$ is $-v$, this then implies that
$(\profx, \profy)$ is not an $\epsilon$-NE, which provides our contradiction.
\qed
\end{proof}

Since Step~\ref{itm:two} implies that the row player's payoff in $(\profx^*,
\profy^*)$ is at least $\frac{2}{3} - z + \epsilon$, Lemma~\ref{lem:apx-min-max}
implies that $\profx^*$ secures a payoff of $\frac{2}{3} - z$ no matter what
strategy the column player plays, which then implies that
Conditions~\ref{itm:pre-one} and~\ref{itm:pre-three} of Lemma~\ref{lem:bs} hold.

Finally, we argue that the algorithm finds a
$(0.6528 + \epsilon)$. The modified Steps~\ref{itm:two} and~\ref{itm:three} now
return a $(\frac{2}{3} - z + \epsilon)$-WSNE, whereas the approximation
guarantees of the other steps are unchanged. Thus, we can reuse our original
analysis to obtain the following theorem.

\begin{theorem}
\label{thm:commbs}
For every $\epsilon > 0$, there is a randomized expected-polynomial-time
algorithm that uses $O\left(\frac{\log^2 n}{\epsilon^2}\right)$ communication
and finds a $(0.6528 + \epsilon)$-WSNE.
\end{theorem}

\paragraph{\bf Query complexity.}
We now show that Algorithm~\ref{alg:ws} can be implemented in a payoff-query
efficient manner. Let $\epsilon > 0$ be a positive constant. We now outline the
changes needed in the algorithm.
\begin{itemize}
\item In Step~\ref{itm:one} we use the algorithm of Lemma~\ref{lem:GRFS} to find
$\frac{\epsilon}{2}$-NEs of $(R, -R)$, and $(C, -C)$. We denote the mixed
strategies found as $(\profx^*_a, \profy^*_a)$ and $(\hat{\profx}_a,
\hat{\profy}_a)$, respectively, and we use these strategies in place of their
original counterparts throughout the rest of the algorithm. We also compute
$\frac{\epsilon}{2}$-approximate payoff vectors for each of these strategies,
and use them whenever we need to know the payoff of a particular strategy under
one of these strategies. In particular, we set $v_r$ to be the payoff of
$\profx^*_a$ according to the approximate payoff vector of $\profy^*_a$, and we
set $v_c$ to be the payoff of $\hat{\profy}_a$ according to the approximate
payoff vector for $\hat{\profx}_a$.

\item In Steps~\ref{itm:two} and~\ref{itm:three} we test against the threshold
of $\frac{2}{3} - z + \epsilon$ rather than $\frac{2}{3} - z$.
\item In Step~\ref{itm:four} we select $j^*$ to be the column that is maximal in
the approximate payoff vector against $\profx^*_a$. We then spend $n$ payoff
queries to query every row in
column $j^*$, which allow us to proceed with the rest of this step as before.
\item In Step~\ref{itm:five} we use the algorithm of Lemma~\ref{lem:GRFS} to
find an approximate payoff vector $v$ for the column player against $\xu$. We
then select $\jp$ to be a column that maximizes $v$, and then spend $n$ payoff
queries to query every row in $\js$, which allows us to proceed with the rest of
this step as before.
\end{itemize}

Observe that the query complexity of the algorithm is $O(\frac{n \cdot \log
n}{\epsilon^2})$, where the dominating term arises due to the use of the
algorithm from Lemma~\ref{lem:GRFS} to approximate solutions to the zero-sum
games.

We now argue that this modified algorithm produces a $(0.6528 + \epsilon)$-WSNE.
Firstly, we need to reestablish the existence of the rows $b$ and $s$ used in
Step~\ref{itm:five}. To do this, we observe that the preconditions of
Lemma~\ref{lem:bs} hold for $\profx^*_a$. We start with
Conditions~\ref{itm:pre-one} and~\ref{itm:pre-three}. Note that the payoff for
the row player under $(\profx^*_a, \profy^*_a)$ is at least $v_r
- \frac{\epsilon}{2}$ (since $v_r$ was estimated with approximate payoff
  vectors,) and Step~\ref{itm:two} ensures that $v_r > \frac{2}{3} - z +
\epsilon$. Hence, we can apply Lemma~\ref{lem:apx-min-max} to argue that
$\profx^*_a$ secures payoff at least $\frac{2}{3} - z$ against every strategy of
the column player, which proves that Conditions~\ref{itm:pre-one}
and~\ref{itm:pre-three} hold. Condition~\ref{itm:pre-two} holds because the
check in Step~\ref{itm:three}, ensures that the approximate payoff of $j^*$
against $\profx^*$ is at least $\frac{2}{3} - z + \epsilon$, and therefore the
actual payoff of of $j^*$ against $\profx^*$ is at least $\frac{2}{3} - z +
\frac{\epsilon}{2}$. Finally, Condition~\ref{itm:pre-four} holds because pure
strategy profiles of this form are explicitly checked for in
Step~\ref{itm:five}.

Steps~\ref{itm:two} and~\ref{itm:three} in the modified algorithm return a
$(\frac{2}{3} - z + \epsilon)$-WSNE, while the other steps provided the same
approximation guarantee as the original algorithm. So, we can reuse the analysis
 for the original algorithm to prove the following theorem.

\begin{theorem}
\label{thm:queries}
There is a randomized algorithm that, with high probability, finds a $(0.6528
+ \epsilon)$-WSNE using $O(\frac{n \cdot \log n}{\epsilon^2})$ payoff queries.
\end{theorem}

\section{A communication efficient algorithm for finding a $\left(\frac{3 - \sqrt{5}}{2} + \epsilon \right)$-NE}

\textbf{The algorithm.} We will study the following algorithm.

\alg{alg:epsne}
\begin{tcolorbox}[title=Algorithm~\ref{alg:epsne}]
\begin{enumerate}
\label{itm:eps-one}
\item Solve the zero-sum games $(R, -R)$ and $(-C, C)$.
\begin{itemize}
\item Let $(\profx^*, \profy^*)$ be
a NE of $(R, -R)$, and let $(\hat{\profx}, \hat{\profy})$ be a NE of $(C, -C)$.
\item
Let $v_r$ be the value secured by $\profx^*$ in $(R, -R)$, and let
$v_c$ be the value secured by $\hat{\profy}$ in $(-C, C)$. Without loss of
generality assume that $v_c \le v_r$.
\item If $v_r \le \frac{3 - \sqrt{5}}{2}$, return $(\hat{\profx}, \profy^*)$.
\end{itemize}
\item Otherwise:
\label{itm:eps-two}
\begin{itemize}
\item Let $j$ be a best response for the column player against $\profx^*$.
\item Let $r$ be a best response for the row player against $j$.
\item Define the strategy profile $\profx' = \frac{1}{2 - v_r} \cdot \profx^* +
\frac{1 - v_r}{2 - v_r} \cdot r$.
\item Return $(\profx', j)$.
\end{itemize}
\end{enumerate}
\end{tcolorbox}

We show that this algorithm always produces a $\frac{3 - \sqrt{5}}{2}$-NE. We
start by considering the strategy profile returned by Step~\ref{itm:eps-one}.
The maximum payoff that the row player can achieve against $\profy^*$ is $v_r$,
so the row player's regret can be at most $v_r$. Similarly, the maximum payoff
that the column layer can achieve against $\hat{\profx}$ is $v_c \le v_r$, so
the column player's regret can be at most $v_r$. Step~\ref{itm:eps-one} only
returns a strategy profile in the case where $v_r \le \frac{3 - \sqrt{5}}{2}$,
so this step always produces a $\frac{3 - \sqrt{5}}{2}$-NE.

To analyse the quality of approximate equilibrium found by
Step~\ref{itm:eps-two}, we use the following Lemma.

\begin{lemma}
\label{lem:epsregret}
The strategy profile $(\profx', j)$ is a $\frac{1 - v_r}{2 - v_r}$-NE.
\end{lemma}
\begin{proof}
We start by analysing the regret of the row player. By definition, row $r$ is a
best response against column $j$. So, the regret of the row player can be
expressed as:
\begin{align*}
R_{rj} - (\profx' \cdot R)_j &=
R_{rj}-\frac{1}{2-v_R} \cdot ((\profx^*)^T \cdot R)_{j}- \frac{1-v_R}{2-v_R} \cdot R_{rj} \\
&\le \frac{1}{2-v_R} \cdot R_{rj} - \frac{1}{2-v_R} \cdot v_R \\
&\le \frac{1}{2-v_R} \cdot 1 - \frac{1}{2-v_R} \cdot v_r \\
&= \frac{1-v_R}{2-v_R},
\end{align*}
where in the first inequality we use the fact that $\profx^*$ is a min-max
strategy that secures payoff at least $v_r$, and the second inequality uses the
fact that $R_{rj} \le 1$.

We now analyse the regret of the column player. Let $c$ be a best response for
the column player against $\profx'$. The regret of the column player can be
expressed as:
\begin{align*}
& \quad\; ((\profx')^T \cdot C)_{c} - ((\profx')^T \cdot C)_{j} \\
&=\frac{1}{2-v_R} \cdot ((\profx^*)^T \cdot C)_{c} + \frac{1-v_R}{2-v_R} \cdot
C_{rc} - \frac{1}{2-v_R} \cdot ((\profx^*)^T \cdot C)_{x^*j} -
\frac{1-v_R}{2-v_R} \cdot C_{rj} \\
&\le \frac{1-v_R}{2-v_R} \cdot C_{rc} - \frac{1-v_R}{2-v_R} \cdot C_{rj} \\
& \le \frac{1-v_R}{2-v_R}.
\end{align*}
The first inequality holds since $j$ is a best response against $x^*$ , and
therefore $((\profx^*)^T \cdot C)_{c} \le (\profx^*)^T \cdot C)_j$, and the
second inequality holds since $C_{rc} \le 1$ and $C_{rj} \ge 0$. Thus, we have
shown that both players have regret at most $\frac{1 - v_r}{2 - v_r}$ under
$(\profx', j)$, and therefore $(\profx', j)$ is a $\frac{1 - v_r}{2 - v_r}$-NE.
\qed
\end{proof}

Step~\ref{itm:eps-two} is only triggered in the case where $v_r >
\frac{3-\sqrt{5}}{2}$, and we have that $\frac{1 - v_r}{2 - v_r} =
\frac{3-\sqrt{5}}{2}$ when $v_r = \frac{3-\sqrt{5}}{2}$. Since $\frac{1 - v_r}{2
- v_r}$ decreases as $v_r$ increases, we therefore have that
Step~\ref{itm:eps-two} always produces a $\frac{3-\sqrt{5}}{2}$-NE. This
completes the proof of correctness for the algorithm.

\paragraph{\bf Communication complexity.} We now argue that, for every $\epsilon
> 0$ the algorithm can be
used to find a $\left(\frac{3 - \sqrt{5}}{2} + \epsilon\right)$-NE using
$O\left(\frac{\log^2 n}{\epsilon^2}\right)$ rounds of communication.

We begin by considering Step~\ref{itm:eps-one}. Obviously, the zero-sum games
can be solved by the two players independently without any communication. Then,
the players exchange $v_r$ and $v_c$ using $O(\log n)$ rounds of communication.
If both $v_r$ and $v_c$ are smaller than $\frac{3 - \sqrt{5}}{2}$, then the
algorithm from Lemma~\ref{lem:GP14} is applied to communicate $\hat{\profx}_s$
to the row player, and $\profy^*_s$ to the column player. Since the payoffs
under the sampled strategies are within $\epsilon$ of the originals, we have
that $(\hat{\profx}_s, \profy^*_s)$ is a $\left(\frac{3 - \sqrt{5}}{2} +
\epsilon \right)$-NE.

If the algorithm reaches Step~\ref{itm:eps-two}, then the row player uses the
algorithm of Lemma~\ref{lem:GP14} to communicate $\profx^*_s$ to the column
player. The column player then computes a best response $j_s$ against
$\profx^*_s$, and uses $\log n$ communication rounds to transmit it to the row
player. The row player then computes a best response $r_s$ against $j_s$, then
computes: $\profx'_s = \frac{1}{2 - v_r} \cdot \profx^*_s + \frac{1 - v_r}{2 -
v_r} \cdot r$, and the players output $(\profx'_s, j_s)$. To see that this
produces a $\left(\frac{3 - \sqrt{5}}{2} + \epsilon \right)$-NE, observe that
$\profx^*_s$ secures a payoff of at least $v_r - \epsilon$ for the row player,
and repeating the proof of Lemma~\ref{lem:epsregret} with this weaker inequality
gives that this strategy profile is a $\left(\frac{1 - v_r}{2 - v_r} + \epsilon
\right)$-NE.

Therefore, we have shown the following theorem.
\begin{theorem}
\label{thm:epscomm}
For every $\epsilon > 0$, there is a randomized expected-polynomial-time
algorithm that uses $O\left(\frac{\log^2 n}{\epsilon^2}\right)$ communication and finds a
$\left(\frac{3 - \sqrt{5}}{2} + \epsilon \right)$-NE.
\end{theorem}

\section{Lower bounds}

Consider the following game.

\begin{center}
\bimatrixgame{3.7mm}{2}{2}{I}{II}%
{{$t$}{$s$}}%
{{$b$}{$r$}}%
{
\payoffpairs{1}{{$0$}{$1$}}{{$1$}{$0.9$}}
\payoffpairs{2}{{$\frac{2}{3}$}{$0.9$}}{{$0$}{$\frac{2}{3}$}}
}
\end{center}

In the game $(R, -R)$, the unique Nash equilibrium is $(b, l)$, which can be
found by iterated elimination of dominated strategies. Similarly, in the game
$(-C, C)$, the unique Nash equilibrium is $(b, r)$, which can again be found by
elimination of dominated strategies. Note, however, that the game itself does
not contain any dominated strategies.
Hence, we have $v_R = v_C = \frac{2}{3}$, so Step~\ref{itm:cfj-two} is triggered,
and the resulting strategy profile is $(b, l)$. Under this strategy profile, the
column player receives payoff $0$, while the best response payoff to the column
player is $\frac{2}{3}$, so this is a $\frac{2}{3}$-WSNE and no better.

This lower bound can be modified to work against our
algorithm for finding a $0.6528$-WSNE by changing both $\frac{2}{3}$ payoffs to
$0.6528$. Then, by the same reasoning given above, Step~\ref{itm:two} is
triggered, and the algorithm returns a $0.6528$-WSNE.

\todo[inline]{This lower bound is not particularly satisfying, because it causes
Step 2 to be triggered, and the rest of the algorithm is ignored. Ideally, for
showing lower bounds, we should recast the algorithm so that all steps are run
in parallel, and the best WSNE found is returned. Work on this is ongoing...}

\bibliographystyle{abbrv}
\bibliography{references}

\newpage
\appendix

\section{A communication efficient algorithm for finding a $0.5$-WSNE in a
win-lose bimatrix game (proof of Theorem~\ref{thm:zo-comm})}
\label{app:zero-one}

We will study the following simple modification of Algorithm~\ref{alg:base}.

\alg{alg:zero-one}
\begin{tcolorbox}[title=Algorithm~\ref{alg:zero-one}]
\begin{enumerate}
\item
\label{itm:zo-one}
Solve the zero-sum games $(R, -R)$ and $(-C, C)$.
\begin{itemize}
\item Let $(\profx^*, \profy^*)$ be
a NE of $(R, -R)$, and let $(\hat{\profx}, \hat{\profy})$ be a NE of $(C, -C)$.
\item
Let $v_r$ be the value secured by $\profx^*$ in $(R, -R)$, and let $v_c$ be the
value secured by $\hat{\profy}$ in $(-C, C)$. Without loss of generality assume
that $v_c \le v_r$.
\end{itemize}
\item
\label{itm:zo-two}
If $v_r \leq 0.5$, then return $(\hat{\profx}, \profy^*)$.
\item
\label{itm:zo-three}
If for all $j \in [n]$ it holds that $C^T_j \cdot \profx^* \leq 0.5$, then
return $(\profx^*, \profy^*)$.
\item
\label{itm:zo-four}
Otherwise:
\begin{itemize}
\item Let $\js$ be a pure best response to $\profx^*$.
\item Find a row $i$ such that $R_{i\js} = 1$ and $C_{ij} = 1$.
\item Return $(i, \js)$.
\end{itemize}
\end{enumerate}
\end{tcolorbox}

We will show that this algorithm always finds a $0.5$-WSNE in a win-lose game.
Firstly, we show that the pure Nash equilibrium found in Step~\ref{itm:zo-four}
always exists. The following lemma is similar to Lemma~\ref{lem:basepure}, but
exploits the fact that the game is win-lose to obtain a stronger conclusion.

\begin{lemma}
\label{lem:zopure}
If Algorithm~\ref{alg:zero-one} is applied to a win-lose game, and it reaches
Step~\ref{itm:zo-four}, then then there exists a row $i \in \supp(\profx^*)$ such that $R_{i\js} =
1$ and $C_{i\js} = 1$.
\end{lemma}
\begin{proof}
Let $i$ be a row sampled from $\profx^*$. We will show that there is a positive
probability that row $i$ satisfies the desired properties.

We begin by showing that the probability that $\pr(R_{i\js} = 0) < 0.5$. Let the
random variable $T = 1 - R_{i\js}$. Since $v_r > \frac{1}{2}$, we have that
$E[T] < 0.5$. Thus, applying Markov's inequality we obtain:
\begin{equation*}
\pr(T \ge 1) \le \frac{E[T]}{1} < 0.5.
\end{equation*}
Since $\pr(R_{i\js} = 0) = \pr(T \ge 1)$ we can therefore conclude that
$\pr(R_{i\js} = 0) < 0.5$. The exact same technique can be used to prove that
$\pr(C_{i\js} = 0) < 0.5$, by using the fact that $C^T_{\js} \cdot \profx^* >
0.5$.

We can now apply the union bound to argue that:
\begin{equation*}
\pr(R_{i\js} = 0 \text{ or } C_{i\js} = 0) < 1.
\end{equation*}
Hence, there is positive probability that row $i$ satisfies $R_{i\js} > 0$ and
$C_{i\js} > 0$, so such a row must exist. The final step is to observe that,
since the game is win-lose, we have that $R_{i\js} > 0$ implies $R_{i\js} = 1$,
and that $C_{i\js} > 0$ implies $C_{i\js} = 1$. \qed
\end{proof}

We now prove that the algorithm always finds a $0.5$-WSNE. The reasoning is very
similar to the analysis of the base algorithm. The strategy profiles returned by
Steps~\ref{itm:zo-two} and~\ref{itm:zo-three} are 0.5-WSNEs by the same
reasoning that was given for the base algorithm. Step~\ref{itm:zo-four} always
returns a pure Nash equilibrium.

\paragraph{\bf Communication complexity.}

We now show that Algorithm~\ref{alg:zero-one} can be implemented in a
communication efficient way.

The zero-sum games in Step~\ref{itm:zo-one} can be solved by the two players
independently without any communication. Then, the players exchange $v_r$ and
$v_s$ using $O(\log n)$ rounds of communication. If both $v_r$ and $v_s$ are
smaller than $0.5$, then the algorithm from Lemma~\ref{lem:GP14} is applied to
communicate $\hat{\profx}_s$ to the row player, and $\profy^*_s$ to the column
player. Since the payoffs under the sampled strategies are within $\epsilon$ of
the originals, we have that all pure strategies have payoff less than or equal
to $0.5 + \epsilon$ under $(\hat{\profx}_s, \profy^*_s)$, so this strategy
profile is a $(0.5 + \epsilon)$-WSNE.

We will assume from now on that $v_r > v_c$. If the algorithm reaches
Step~\ref{itm:zo-three}, then the row player uses the algorithm of
Lemma~\ref{lem:GP14} to communicate $\profx^*_s$ to the column player. The
column player then computes a best response $\js_s$ against $\profx^*_s$, and
checks whether the payoff of $\js_s$ against $\profx^*_s$ is less than or equal
to $0.5 + \epsilon$. If so, then the players output $(\profx^*_s, \js_s)$, which
is a $0.5 + \epsilon$-WSNE

Otherwise, we claim that there is a pure strategy $i \in \supp(\profx^*_s)$ such
that $(i, \js_s)$ is a pure Nash equilibrium. This can be shown by observing
that the expected payoff of $\profx^*_s$ against $\js_s$ is at least $0.5 -
\epsilon$, while the expected payoff of $\js_s$ against $\profx^*_s$ is at least
$0.5 + \epsilon$. Repeating the proof of Lemma~\ref{lem:zopure} using these
inequalities then shows that the pure Nash equilibrium does indeed exist. Since
$\supp(\profx^*_s)$ has logarithmic size, the row player can simply transmit to
the column player all payoffs $R_{i\js_s}$ for which $i \in \supp(\profx^*_s)$,
and the column player can then send back a row corresponding to a pure Nash
equilibrium.

In conclusion, we have shown that a $(0.5 + \epsilon)$-WSNE can be found in
randomized expected-polynomial-time using $O\left(\frac{\log^2
n}{\epsilon^2}\right)$ communication, which completes the proof of
Theorem~\ref{thm:zo-comm}.

\section{Proof of Lemma~\ref{lem:bs}}
\label{app:bs}

In this section we assume that Steps~\ref{itm:one} through~\ref{itm:four} of our
algorithm did not return a $(\frac{2}{3} - z)$-WSNE, and that neither $\js$ nor
$\jp$ contained a pure $(\frac{2}{3} - z)$-WSNE. We show that, under these
assumptions, the rows $b$ and $s$ required by Step~\ref{itm:five} do indeed
exist.

\paragraph{\bf Probability bounds.}

We begin by proving bounds on the amount of probability that $\xs$ can place on
$\sr$ and $\brs$. The following lemma uses the fact that $\xs$ secures an
expected payoff of at least $\frac{2}{3} - z$ to give an upper bound on the
amount of probability that $\xs$ can place on $\sr$. To simplify notation, we
use $\pr(\brs)$ to denote the probability assigned by $\xs$ to the rows in
$\brs$, and we use $\pr(\sr)$ to denote the probability assigned by $\xs$ to the
rows in $\sr$.

\begin{lemma}
\label{lem:sprob}
$\pr(\sr) \leq \frac{1+3z}{2-3z}$.
\end{lemma}
\begin{proof}
We will prove our claim using Markov's inequality. Consider the random variable
$T = 1 - R_{i\js}$ where $i$ is sampled from \xs. Since by our assumption
the expected payoff of the row player is greater than $2/3 -z$ we get that
$E(T) \leq 1/3 + z$. If we apply Markov's inequality we get
\begin{align*}
Pr(T \geq \frac{2}{3} - z) \leq \frac{E(T)}{\frac{2}{3} - z} \leq \frac{1+3z}{2-3z}
\end{align*}
which is the claimed result.
\qed
\end{proof}

Next we show an upper bound on $\pr(\brs)$. Here we use the fact that $\js$ does
not contain a $(\frac{2}{3} - z)$-WSNE to argue that all column player payoffs
in $\brs$ are smaller than $\frac{1}{3} + z$. Since we know that the payoff of
$\js$ against $\xs$ is at least $\frac{2}{3} -z$, this can be used to prove a
upper bound on the amount of probability that $\xs$ assigns to $\brs$.

\begin{lemma}
\label{lem:bprob}
$\pr(\brs) \leq \frac{1+3z}{2-3z}$.
\end{lemma}
\begin{proof}
Since there is no $i \in \supp(\xs)$ such that $(i, \js)$ is a pure
$(\frac{2}{3} - z)$-WSNE , and since each row $i \in \brs$ satisfies $R_{i\js}
\ge \frac{1}{3} + z$, we must have that $C_{i\js} < \frac{1}{3} + z$ for every
$i \in B$. By assumption we know that $C^T_{\js}
\xs > 2/3 -z$. So, we have the following inequality:
\begin{equation*}
\frac{2}{3} - z < \pr(\brs) \cdot (\frac{1}{3} + z) + \bigl(1- \pr(\brs)\bigr) \cdot 1.
\end{equation*}
Solving this inequality for $\pr(\brs)$ gives the desired result.
\qed
\end{proof}

\paragraph{\bf Payoff inequalities for $\js$.}

We now show properties about the average payoff obtained from the rows in $\brs$
and $\sr$. Recall that $\xu$ was defined in Step~\ref{itm:four} of our
algorithm, and that it denotes the normalization of the probability mass
assigned by $\xs$ to rows in $\brs$. The following lemma shows that the expected
payoff to the row player in the strategy profile $(\xu, \js)$ is close to $1$.

\begin{lemma}
\label{lem:brjs}
We have $(\xu^T \cdot R)_{\js} > \frac{1 - 6z}{1 + 3z}$.
\end{lemma}
\begin{proof}
By definition we have that:
\begin{equation}
\label{eqn:bjrs}
(\xu^T \cdot R)_{\js} = \frac{1}{\pr(\brs)} \cdot
\sum_{i \in \brs} \xs_i \cdot R_{i\js}.
\end{equation}
We begin by deriving a lower bound for $\sum_{i \in \brs} \xs_i \cdot R_{i\js}$.
Using the fact that $\xs$ secures an expected payoff of at least $2/3 -z$
against $\js$ and
then applying the bound from Lemma~\ref{lem:sprob} gives:
\begin{align*}
\frac{2}{3}-z & <
\sum_{i \in \brs} \xs_i \cdot R_{i\js} + (\frac{1}{3}+z) \cdot \pr(\sr) \\
& \leq \sum_{i \in \brs} \xs_i \cdot R_{i\js} + (\frac{1}{3}+z)\cdot \frac{1+3z}{2-3z}.
\end{align*}
Hence we can conclude that:
\begin{align*}
\sum_{i \in \brs} \xs_i \cdot R_{i\js} &> \frac{2}{3}-z - \frac{1}{3} \cdot
\frac{(1+3z)^2}{2-3z} \\
& = \frac{1 - 6z}{2 - 3z}.
\end{align*}
Substituting this into Equation~\eqref{eqn:bjrs}, along with the upper bound
on $\pr(\brs)$ from Lemma~\ref{lem:bprob}, allows us to conclude that:
\begin{align*}
(\xu^T \cdot R)_{\js} &\ge \frac{2 - 3z}{1+3z} \cdot
\sum_{i \in \brs} \xs_i \cdot R_{i\js} \\
& > \frac{2 - 3z}{1+3z} \cdot
\frac{1 - 6z}{2 - 3z} \\
& = \frac{1 - 6z}{1 + 3z}.
\end{align*}
\qed
\end{proof}

Next we would like to show a similar bound on the expected payoff to the column
player of the rows in $\sr$. To do this, we define $\xl$ to be the normalisation
of the probability mass that $\xs$ assigns to the rows in $\sr$. More formally,
for each $i \in [n]$, we define:
\begin{equation*}
(\xl)_i = \begin{cases}
\frac{1}{\pr(\sr)} \cdot \xs_i & \text{if $i \in \sr$} \\
0 & \text{otherwise.}
\end{cases}
\end{equation*}
The next lemma shows that the expected payoff to the column player in the
profile $(\xl, \js)$ is close to $1$.

\begin{lemma}
\label{lem:scjs}
We have $(\xl^T \cdot C)_{\js} > \frac{1 - 6z}{1 + 3z}.$
\end{lemma}
\begin{proof}
By definition we have that:
\begin{equation}
\label{eqn:scjs}
(\xl^T \cdot C)_{\js} = \frac{1}{\pr(\sr)} \cdot
\sum_{i \in \sr} \xs_i \cdot C_{i\js}.
\end{equation}
We begin by deriving a lower bound for $\sum_{i \in \sr} \xs_i \cdot C_{i\js}$.
By assumption, we know that $C^T_{\js}\xs > 2/3 -z$. Moreover, since $\js$ does
not contain a $(\frac{2}{3} - z)$-WSNE we have that all rows $i$ in $\brs$
satisfy $C_{i\js} < 1/3-z$. If we combine these facts that with
Lemma~\ref{lem:bprob} we obtain:
\begin{align*}
\frac{2}{3}-z & <
\sum_{i\in \sr} \xs_i \cdot C_{i \js} + (\frac{1}{3}+z)\cdot \pr(\brs)\\
& \leq
\sum_{i\in \sr} \xs_i \cdot C_{i \js} + (\frac{1}{3}+z)\cdot \frac{1+3z}{2-3z}.
\end{align*}
Hence we can conclude that:
\begin{align*}
\sum_{i\in \sr} \xs_i \cdot C_{i \js} &> \frac{2}{3}-z - \frac{1}{3} \cdot
\frac{(1+3z)^2}{2-3z} \\
& = \frac{1 - 6z}{2 - 3z}.
\end{align*}
Substituting this into Equation~\eqref{eqn:scjs}, along with the upper bound
on $\pr(\sr)$ from Lemma~\ref{lem:bprob}, allows us to conclude that:
\begin{align*}
(\xu^T \cdot R)_{\js} &\ge \frac{2 - 3z}{1+3z} \cdot
\sum_{i \in \brs} \xs_i \cdot R_{i\js} \\
& > \frac{2 - 3z}{1+3z} \cdot
\frac{1 - 6z}{2 - 3z} \\
& = \frac{1 - 6z}{1 + 3z}.
\end{align*}
\qed
\end{proof}

\paragraph{\bf Payoff inequalities for $\jp$.}

We now want to prove similar inequalities for the column $\jp$. The next lemma
shows that the expected payoff for the column player in the profile $(\xu, \jp)$
is close to $1$. This is achieved by first showing a lower bound on the payoff
to the column player in the profile $(\xu, \js)$, and then using the fact that
$\js$ is not a $(\frac{2}{3} - z)$-best-response against $\xu$, and that $\jp$
is a best response against $\xu$.

\begin{lemma}
\label{lem:bcjp}
We have $(\xu^T \cdot C)_{\jp} > \frac{1 - 6z}{1 + 3z}$.
\end{lemma}
\begin{proof}
We first establish a lower bound on $(\xu^T \cdot C)_{\js}$. By assumption, we
know that $C^T_{\js} \xs > 2/3 -z$. Using this fact, along with the bounds from
Lemmas~\ref{lem:sprob} and~\ref{lem:bprob} gives:
\begin{align*}
\frac{2}{3} - z &< \pr(\brs) \cdot (\xu^T \cdot C)_{\js} + \pr(\sr) \cdot 1\\
& \le \frac{1+3z}{2-3z} \cdot (\xu^T \cdot C)_{\js} + \frac{1+3z}{2-3z}.
\end{align*}
Solving this inequality for $(\xu^T \cdot C)_{\js}$ yields:
\begin{equation*}
(\xu^T \cdot C)_{\js} > \frac{1}{3} \cdot \frac{1 - 21z + 9 z^2}{1 + 3z}.
\end{equation*}

Now we can prove the lower bound on $(\xu^T \cdot C)_{\jp}$. Since $\js$ is not
a $(\frac{2}{3} - z)$-best-response against $\xu$, and since $\jp$ is a best
response against $\xu$ we obtain:
\begin{align*}
(\xu^T \cdot C)_{\jp} &> (\xu^T \cdot C)_{\js} + \frac{2}{3} - z \\
(\xu^T \cdot C)_{\jp} & > \frac{1}{3} \cdot \frac{1 - 21z + 9 z^2}{1 + 3z} +
\frac{2}{3} - z \\
&= \frac{1 - 6z}{1 + 3z}.
\end{align*}
\qed
\end{proof}

The only remaining inequality that we require is a lower bound on the expected
payoff to the row player in the profile $(\xl, \jp)$. However, before we can do
this, we must first prove an upper bound on the expected payoff to the row
player in $(\xu, \jp)$, which we do in the following lemma. Here we first prove
that most of the probability mass of $\xu$ is placed on rows $i$ in which
$C_{i\jp} > \frac{1}{3} + z$, which when combined with the fact that there is no
$i \in \supp(\xs)$ such that $(i, \jp)$ is a pure $(\frac{2}{3} - z)$-WSNE, is
sufficient to provide an upper bound.

\begin{lemma}
\label{lem:brjp-upper}
We have $(\xu^T \cdot R)_{\jp} < \frac{1}{3} \cdot \frac{1 + 33z + 9 z^2}{1 + 3z}.$
\end{lemma}
\begin{proof}
We begin by proving an upper bound on the
amount of probability mass assigned by $\xu$ to rows $i$ with $C_{i\jp} <
\frac{1}{3} + z$.
Let $T = 1 - C_{i\jp}$ be a random variable where the row $i$ is sampled
according to $\xu$. Lemma~\ref{lem:bcjp} implies that:
\begin{equation*}
E[T] = 1 - \frac{1 - 6z}{1 + 3z} = \frac{9z}{1+3z}.
\end{equation*}
Observe that $\pr(T \ge 1 - (\frac{1}{3} + z)) = \pr(T \ge \frac{2}{3} - z)$ is
equal to the
amount of probability that
$\xu$ assigns to rows $i$ with $C_{i\jp} < \frac{1}{3} + z$. Applying Markov's
inequality then establishes our bound.
\begin{equation*}
\pr(T \ge \frac{2}{3} - z) \le \frac{\frac{9z}{1+3z}}{\frac{2}{3} -z}. \\
\end{equation*}

So, if $p = \frac{9z}{(1+3z)(2/3 -z)}$ then we know that at least $1-p$
probability is assigned by $\xu$ to rows $i$ such that $C_{i\jp}> \frac{1}{3} +
z$. Since we have assumed that there is no $i \in \supp(\xs)$ such that $(i,
\jp)$ is a pure $(\frac{2}{3} - z)$-WSNE, we know that any such row $i$ must
satisfy $R_{i\jp} < \frac{1}{3} + z$. Hence, we obtain the following bound:
\begin{align*}
(\xu^T \cdot R)_{\jp} &< (1-p) \cdot (\frac{1}{3} + z) + p \\
& = \frac{1}{3} \cdot \frac{1 + 33z + 9 z^2}{1 + 3z}.
\end{align*}
\qed
\end{proof}

Finally, we show that the expected payoff to the row player in the profile
$(\xl, \jp)$ is close to $1$. Here we use the fact that $\xs$ is a min-max
strategy along with the bound from Lemma~\ref{lem:brjp-upper} to prove our lower
bound.

\begin{lemma}
\label{lem:srjp}
We have $(\xl^T \cdot R)_{\jp} > \frac{1 - 15z}{1+3z}$.
\end{lemma}
\begin{proof}
Since $\xs$ is a min-max strategy that secures a value strictly larger than
$\frac{2}{3} - z$, we have:
\begin{equation*}
\frac{2}{3} - z < \pr(\brs) \cdot (\xu^T \cdot R)_{\jp} +
\pr(\sr) \cdot (\xl^T \cdot R)_{\jp}.
\end{equation*}
Substituting the bounds from
Lemmas~\ref{lem:sprob},~\ref{lem:bprob},~and~\ref{lem:brjp-upper} then gives:
\begin{align*}
\frac{2}{3} - z <
\frac{1+3z}{2-3z} \cdot
\frac{1}{3} \cdot \frac{1 + 33z + 9 z^2}{1 + 3z} +
\frac{1+3z}{2-3z} \cdot (\xl^T \cdot R)_{\jp}.
\end{align*}
Solving for $(\xl^T \cdot R)_{\jp}$ then yields the desired result.
\qed
\end{proof}

\paragraph{\bf Finding rows $b$ and $u$.}

So far, we have shown that the expected payoff to the row player in $(\xu, \js)$
is close to $1$, and that the expected payoff to the column player in $(\xu,
\jp)$ is close to $1$. We now show that there exists a row $b \in \brs$ such
that $R_{b\js}$ is close to $1$, and $C_{b\jp}$ is close to $1$, and that there
exists a row $s \in \sr$ in which $C_{s\js}$ and $R_{s\jp}$ are both close to
$1$.  The following lemma uses Markov's inequality to show a pair of probability
bounds that will be critical in showing the existence of $b$.

\begin{lemma}
\label{lem:b-prob}
We have:
\begin{itemize}
\item  $\xu$ assigns strictly more than $0.5$ probability to rows $i$ with $R_{i\js} > 1- \frac{18z}{1 + 3z}$.
\item $\xu$ assigns strictly more than $0.5$ probability to rows
$i$ with $C_{i\jp} > 1- \frac{18z}{1 + 3z}$.
\end{itemize}
\end{lemma}
\begin{proof}
We begin with the first case. Consider the random variable $T = 1 - R_{i\js}$
where $i$ is sampled from \xu. By Lemma~\ref{lem:brjs}, we have that:
\begin{equation*}
E[T] < 1 - \frac{1 - 6z}{1 + 3z} = \frac{9z}{1 + 3z}.
\end{equation*}
We have that $T \ge \frac{18z}{1 + 3z}$ whenever $R_{i\js} \le 1 -
\frac{18z}{1 + 3z}$, so we can apply Markov's inequality to obtain:
\begin{equation*}
\pr(T \ge \frac{18z}{1 + 3z}) < \frac{\frac{9z}{1 + 3z}}{\frac{18z}{1 + 3z}} =
0.5.
\end{equation*}

The proof of the second case is identical to the proof given above, but uses the
(identical) bound from Lemma~\ref{lem:bcjp}.
\qed
\end{proof}

The next lemma uses the same techniques to prove a pair of probability bounds
that will be used to prove the existence of $s$.
\begin{lemma}
\label{lem:s-prob}
We have:
\begin{itemize}
\item $\xl$ assigns strictly more than $\frac{1}{3}$ probability to
rows $i$ with $C_{i\js} > 1- \frac{27z}{1 + 3z}$.
\item
$\xl$ assigns strictly more than $\frac{2}{3}$ probability to
rows $i$ with $R_{i\jp} > 1 - \frac{27z}{1 + 3z}$.
\end{itemize}
\end{lemma}
\begin{proof}
We begin with the first claim. Consider the random variable $T = 1 - C_{i\js}$
where $i$ is sampled from \xl. By Lemma~\ref{lem:scjs}, we have that:
\begin{equation*}
E[T] < 1 - \frac{1 - 6z}{1 + 3z} = \frac{9z}{1 + 3z}.
\end{equation*}
We have that $T \ge \frac{27z}{1 + 3z}$ whenever $C_{i\js} \le 1 -
\frac{27z}{1 + 3z}$, so we can apply Markov's inequality to obtain:
\begin{equation*}
\pr(T \ge \frac{27z}{1 + 3z}) < \frac{\frac{9z}{1 + 3z}}{\frac{27z}{1 + 3z}} =
\frac{1}{3}.
\end{equation*}

We now move on to the second claim. Consider the random variable $T = 1 -
R_{i\js}$ where $i$ is sampled from \xu. By Lemma~\ref{lem:srjp}, we have that:
\begin{equation*}
E[T] < 1 - \frac{1 - 15z}{1 + 3z} = \frac{18z}{1 + 3z}.
\end{equation*}
We have that $T \ge \frac{27z}{1 + 3z}$ whenever $R_{i\js} \le 1 -
\frac{27z}{1 + 3z}$, so we can apply Markov's inequality to obtain:
\begin{equation*}
\pr(T \ge \frac{27z}{1 + 3z}) < \frac{\frac{18z}{1 + 3z}}{\frac{27z}{1 + 3z}} =
\frac{2}{3}.
\end{equation*}
\qed
\end{proof}

Finally, we can formally prove the existence of $b$ and $s$, which completes the
proof of correctness for our algorithm.

\begin{proof}[of Lemma~\ref{lem:bs}]
We begin by proving the first claim. If we sample a row $b$ randomly from $\xu$,
then Lemma~\ref{lem:b-prob} implies that probability that $R_{b\js} \le 1-
\frac{18z}{1 + 3z}$ is strictly less than $0.5$ and that the probability that
$C_{b\jp} \le 1- \frac{18z}{1 + 3z}$ is strictly less than $0.5$. Hence, by the
union bound, the probability that at least one of these events occurs is
strictly less than $1$. So, there is a positive probability that neither of the
events occurs, which implies that there exists at least one row $b$ that
satisfies the desired properties.

The second claim is proved using exactly the same technique, but using the
bounds from Lemma~\ref{lem:s-prob}, again observing that the probability that a
randomly sampled row from $\xl$ satisfies the desired properties with positive
probability. \qed
\end{proof}

This completes the proof of Lemma~\ref{lem:bs}.

\end{document}